\documentclass[12pt,a4paper]{amsart}
\usepackage{amsmath}
\usepackage{amsfonts}
\usepackage{amssymb}
\usepackage{amsthm}
\usepackage{graphicx}

\newtheorem{theorem}{Theorem}
\newtheorem{proposition}[theorem]{Proposition}
\theoremstyle{definition}
\newtheorem{definition}[theorem]{Definition}

\def\N{\mathbb{N}}

\begin{document}

\title{Quantum logics close to Boolean algebras}

\author{Mirko Navara}

\author{Pavel Pt\'{a}k}

\begin{abstract}
We consider orthomodular posets endowed with a symmetric difference. We call them ODPs. 
Expressed in the quantum logic language, we consider quantum logics with an XOR-type connective.
We study three classes of ``almost Boolean'' ODPs, two of them defined by requiring rather specific behaviour of infima and the third by a Boolean-like behaviour of Frink ideals.
We establish a (rather surprising) inclusion between the three classes, shadding thus light on their intrinsic properties.
(More details can be found in the Introduction that follows.)
Let us only note that the orthomodular posets pursued here, though close to Boolean algebras (i.e., close to 
standard 
quantum logics), 
still have a potential for an arbitrarily high degree of non-compatibility and hence they may enrich the studies of mathematical foundations of quantum mechanics.
\end{abstract}

\maketitle
\footnotetext[1]{Mirko Navara (corresponding author), 
Department of Cybernetics,
Faculty of Electrical Engineering,
Czech Technical University in Prague,
Czech Republic,
\email{navara@cmp.felk.cvut.cz}}

\footnotetext[2]{Pavel Pt\'{a}k,
Department of Mathematics,
Faculty of Electrical Engineering,
Czech Technical University in Prague,
Czech Republic,
\email{ptak@math.feld.cvut.cz}}

\noindent {\bf Keywords:} Quantum logic with a symmetric difference, Frink ideal, Boolean algebra.

\noindent {\bf AMS Classification:} 06C15, 03G12, 81P10

\section{Introduction}

In certain quantum axiomatics the events of a quantum experiment are associated with an orthomodular partially ordered set.
This orthomodular partially ordered set is then called a quantum logic (QL).
Traditionally, the QLs were assumed to be projections in a Hilbert space or QLs of a similar type
\cite{Birkhoff-vonNeumann,Gleason,Mackey}, later some more general QLs were investigated
\cite{DvurPulm,Gudder:book,PP,Redei}.
In this note we adopt a kind of ``antiprojection'' approach---we restrict our attention to QLs that are as close as possible to the standard ones.
Let us formally introduce them.

Let us first assume that the QLs to be considered allow for an introduction of a natural symmetric difference.
Let us call them ODPs (see~\cite{MatousekP:Order}).
Going further towards Boolean algebras, let us introduce the following two classes:
Let us denote by $\mathcal{R}$ (resp.~$\mathcal{T}$) the class of all ODPs that are determined as follows:
$P\in\mathcal{R}\iff P$ is na ODP and the implication 
$a\land b=0 \implies a\le b^\perp$ holds true
(resp.~$P\in\mathcal{T}\iff P$ is na ODP and the implication 
$a\land b=a\land b^\perp=0 \implies a\le b^\perp$ holds true).
Consider also the class $\mathcal{S}$ of all ODPs in which all maximal Frink ideals are selective
(a well-known
We then show that $\mathcal{R}\subset\mathcal{S}\subset\mathcal{T}$
and that $\mathcal{R}\ne\mathcal{T}$.
We also comment on other properties of the classes.

\section{Notations used throughout the paper}

By an orthocomplemented poset we mean a pentuple $(X,\le,\,^\perp,0,1)$
such that $\le$ is a partial ordering on the set~$X$ with a least (resp.\ greatest) element~$0$ (resp.~$1$)
and with the unary operation $\,^\perp$ on~$X$ such that, for each $x,y\in X$, we have
$x\land x^\perp=0$, $x\lor x^\perp=1$, $(x^\perp)^\perp=x$, and $x\le y$ implies $y^\perp\le x^\perp$.

Our principal definition reads as follows.

\begin{definition}
Let $P=(X,\le,\,^\perp,0,1,\bigtriangleup)$, where $(X,\le,\,^\perp,0,1)$ is an orthocomplemented poset and 
$\bigtriangleup\colon X^2\to X$ is a binary operation.
Then $P$ is said to be an \emph{orthocomplemented difference poset} (ODP) if the following three conditions are fulfilled for all $x,y,z\in P$:
\begin{enumerate}
\item[(Def~1)] $x\bigtriangleup(y\bigtriangleup z)=(x\bigtriangleup y)\bigtriangleup z$,
\item[(Def~2)] $x\bigtriangleup 1=1\bigtriangleup x=x^\perp$,
\item[(Def~3)] $x\le z, y\le z \implies x\bigtriangleup y\le z$.
\end{enumerate}
\end{definition}

It can be proved that each ODP is orthomodular. 
To do that, suppose that $P$ is an ODP and $x,y\in P$ with $x \le y$. 
We have to check the identity $y=x \lor (y \land x^\perp)$. A simple calculation gives us
$y \land (x \lor (y \land x^\perp) )^\perp =  y \land x^\perp \land (y \land x^\perp)^\perp  = 0$. 
But this means in ODPs that $y = x \lor (y \land x^\perp)$  alias this ensures the orthomodularity. 
To check the latter write $z = x \lor (y \land x^\perp)$. We have $z \le y$  and  $y \land z^\perp=0$. 
We see 
that $y \land z^\perp = y \bigtriangleup z^\perp = 0$ . Further, 
$$z = z \bigtriangleup 0 =  z \bigtriangleup (y^\perp \bigtriangleup y^\perp) = (z \bigtriangleup y^\perp) \bigtriangleup y^\perp = 0 \bigtriangleup y^\perp = y\,.$$ 
So $P$ is orthomodular and hence the ODPs are a kind of enriched quantum logics. (We shall study certain algebraic properties of ODPs; the state properties have been investigated in 
\cite{HroPta} and~\cite{MatousekP:Order}).

Each Boolean algebra is an ODP, of course.
A ``proper'' example of an ODP is subsets of even cardinalities.
Formally, let $k\in\N$, $\Omega=\{1,2,\ldots, 2k-1, 2k\}$.
Let $X$ be the collection of all subsets of~$\Omega$ that consist of an even number of elements. 
Then $X$ with the inclusion relation for $\le$ and the complement operation in~$\Omega$ for~$^\perp$ is an ODP.

We shall also deal with a type of ideals in ODPs.
Suppose that $P$ is an ODP and $A\subset P$.
Let us write $A^\uparrow=\{x\in P\mid y\le x$ for any $y\in A\}$
and $A^\downarrow=\{x\in P\mid x\le y$ for any $y\in A\}$.
Let us further write $A^{\uparrow\downarrow}$ for $(A^\uparrow)^\downarrow$.

\begin{definition} (see~\cite{Frink})
Suppose that $P$ is an ODP and $I\subset P$.
We say that $I$ is a \emph{Frink ideal} if
(1)~$1\in I$, and
(2)~for any finite~$J$, $J\subset I$, we have $J^{\uparrow\downarrow}\subset I$.
Further, the ideal $I$ is said to be \emph{selective}
if for any pair $\{a,a^\perp\}$ either $a\in I$ or $a^\perp\in I$.
(Thus, in other words, $\operatorname{card}(\{a,a^\perp\}\cap I=1$.)
\end{definition}

The following result is plausible.

\begin{proposition}
(1)~Each Frink ideal can be extended to a maximal Frink ideal.

(2)~A selective Frink ideal is maximal.
\end{proposition}

\begin{proof}
Part~(1) is a routine consequence of Zorn's lemma.
Part~(2) can be easily checked:
If $p\in J\setminus I$ for Frink ideals $I,J$,
then neither $p$ nor $p^\perp$ belongs to~$I$.
\end{proof}

At this moment we are ready for our investigation.
Having defined the classes $\mathcal{R}$ and~$\mathcal{T}$ in the Introduction,
we add another class of ODPs to them.
Let us denote by $\mathcal{S}$ the class of all ODPs such that 
$P\in \mathcal{S}$ if each maximal Frink ideal in $P$ is selective.

\section{Results}

We will compare the classes $\mathcal{R}$, $\mathcal{S}$, and $\mathcal{T}$.
The main result reads as follows.
Prior to its formulation, let us note that the proof strategy follows that of~\cite{NP:Frink};
the presence of $\bigtriangleup$ and the pecularities of the classes $\mathcal{R}$, $\mathcal{S}$, and $\mathcal{T}$
require some additional checking in places.

\begin{theorem}
 $\mathcal{R}\subset\mathcal{S}\subset\mathcal{T}$.
\end{theorem}

\begin{proof}
 Trying to show $\mathcal{R}\subset\mathcal{S}$, suppose that $P\in\mathcal{R}$ 
 and $I$ is a maximal Frink ideal in~$P$.
 Suppose further that $a\notin I$.
 We want to show that $a^\perp\in I$.
 It is evident that if $F\subset I$, $F$~finite, then $a\notin F^{\uparrow\downarrow}$.
 As a consequence, there is such an upper bound $b_F$ of~$F$ that $a\not\le b_F$.
 Since $(b_F^\perp)^\perp=b_F$, we have $a\not\le(b_F^\perp)^\perp$ and therefore 
 (as $P\in\mathcal{R}$),
 $a\land b_F^\perp\ne0$.
 This implies that there is an element $c_F$ such that $c_F\ne0$
 and $c_F\le a, c_F\le b_F^\perp$.
 Hence $c_F^\perp\ne1$ and $c_F^\perp$ is an upper bound of $F\cup a^\perp$.
 But then the set $G=\bigcup_F (F\cup a^\perp)^{\uparrow\downarrow}$, where one takes the union over all finite subsets $F$ of~$I$,
 is a Frink ideal that extends~$I$. Since $I$ is maximal, we infer that $I=G$ and hence $a^\perp\in I$. 
 So $\mathcal{R}\subset\mathcal{S}$.
 
 Let us show that $\mathcal{S}\subset\mathcal{T}$.
 Suppose that $P\in\mathcal{S}$ and suppose further that $a\land b=a\land b^\perp=0$ ($a,b\in P$).
 Assume that $a\ne0$.
 Consider the Frink ideal $J=\{p\in P\mid p\le a^\perp\}
 =\{a^\perp\}^\downarrow$.
 Let $I$ be a maximal Frink ideal that contains~$J$.
 Making use of the De~Morgan law, we have $a^\perp\lor b^\perp=a^\perp\lor b=1$.
 This implies that neither $b^\perp$ nor $b$ belongs to~$I$.
 Indeed, since $a^\perp\in I$, we would infer that $1\in I$ and that is excluded.
 So $P\in\mathcal{T}$ and this completes the proof.
\end{proof}

Let us make a few remarks on the classes investigated. 
First, it is obvious that we pursue a ``properly non-lattice'' matter.
Indeed, if $P\in\mathcal{T}$ and $P$ is a lattice,
then $P$ is easily seen to be a Boolean algebra (see also~\cite{Tkadlec:TMMP}).

Second, if $P\in\mathcal{R}$ then $P$ is set-representable as an ODP.
So $P$ can be represented as a collection of subsets of a set with the inclusion for $\le$ and the set-complement for~$^\perp$,
with the presence of the set symmetric difference.
This has been proved in~\cite{MatousekP:Order} by considering the two-valued $\bigtriangleup$-states on~$P$.
In the ``ideal'' reformulation \`a~la Boolean algebras, the ODP set representation of~$P$ follows from the existence of the order-determining collection of maximal Frink ideals in~$P$.
Indeed, observe that a Frink ideal is closed under the operation $\bigtriangleup$ (Def~3).
Then the inequality $a\not\le b$ implies $a\land b^\perp\ne0$ and therefore there is an element
$c$, $c\ne0$, with $c\le a$ and $c\le b^\perp$.
Since $\mathcal{R}\subset\mathcal{S}$, there is a maximal Frink ideal, $I$, that extends the ideal
$J=\{p\in P\mid p\le c\}=\{c\}^\downarrow$
and that is selective.
So we infer that $a\in I$ and $b\notin I$.
By the standard Boolean-like reasoning we see that $P$~is set-representable as an ODP
(in fact, if $Q$ is the collection of all selective Frink ideals on~$P$ then a~set representation of~$P$ can be obtained as a collection of subsets of~$Q$).

Third, let us indicate a typical example of an ODP that belongs to~$\mathcal{R}$ and, also, let us see why $\mathcal{R}$ is strictly smaller than~$\mathcal{T}$.
Let $\N$ be the set of all natural numbers.
Let us set, for each $i\in\{0,1,2,3,4,5\}$, $N_i=\{n\in\N\mid n=6k+i$ for some natural~$k\}$.
Let us further denote by $A_1$ the set of all even numbers. Now, write
$A_2=N_0\cup N_3$, $A_3=N_2\cup N_3\cup N_4$ ($=A_1\bigtriangleup A_2$).
Consider the collection
$$\widetilde{R}=\{\emptyset, A_1, A_2, A_3,  A_1^\perp, A_2^\perp, A_3^\perp, \N\}$$
and set 
$$R=\{X\subset\N\mid X\bigtriangleup D \text{ is finite for some } D\in\widetilde{R}\}\,.$$
Then one easily shows that $R\in\mathcal{R}$ ($R$~is obviously not Boolean).
Note on dealing with~$R$ that, since any cartesian product of ODPs is an ODP, we can easily construct an ODP 
with an arbitrarily high degree of non-compatibility (a certain virtue of~$\mathcal{R}$ in a possible application of $\mathcal{R}$ in the quantum axiomatics).
Another property of $\mathcal{R}$ worth observing in connection with a potential interpretation
in a quantum theory is the fact that $a,b\in\mathcal{R}$ are compatible exactly when $a \land b$ exists---an ``algebraic'' test for compatibility.

In order to find an ODP that does belong to~$\mathcal{T}$ and does \emph{not} belong to~$\mathcal{R}$,
let us take the above collection~$\widetilde{R}$ and add to it all singletons of the set 
$E=(A_1\cap A_2^\perp) \cup (A_2\cap A_3)$.
Let us denote by $\widetilde{T}$ the so obtained collection and let us denote by $T$ the ODP generated by $\widetilde{T}$
(it is sufficient to add symmetric differences and complements in an obvious way).
Then $T\in\mathcal{T}$ and $T\notin\mathcal{R}$.

\bigskip
\textbf{Acknowledgements}. The first author was supported by the Czech Science Foundation grant 20-09869L.
The second author was supported by the European Regional Development Fund, project ``Center for Advanced Applied Science'' (No.\ CZ.02.1.01/0.0/0.0/16\_019\allowbreak/0000778).

\end{document}